\documentclass[11pt]{article}
\usepackage[top=1in,bottom=1in,left=1in,right=1in]{geometry}
\PassOptionsToPackage{numbers, compress}{natbib}

\usepackage[utf8]{inputenc} 
\usepackage[T1]{fontenc}    
\usepackage{hyperref}       
\usepackage{url}            
\usepackage{booktabs}       
\usepackage{amsfonts}       
\usepackage{nicefrac}       
\usepackage{microtype}      
\usepackage{xcolor}         
\usepackage{algorithm}
\usepackage{algpseudocode}
\usepackage{amsmath}
\usepackage{svg}
\usepackage{algorithmicx}
\usepackage{tikz}
\usepackage{subfig}
\usepackage{wrapfig} 
\usepackage{amsmath}
\usepackage[sets,operators]{cryptocode}
\usepackage{amssymb}
\usepackage{authblk}
\usepackage{dashbox}
\usepackage{tikz-qtree}
\usepackage{tikz-qtree-compat}
\usetikzlibrary{matrix}
\usepackage{tabularx}
\usepackage{minitoc}
\usepackage[parfill]{parskip}
\usepackage{extarrows}

\title{SLIP-SEC: Formalizing Secure Protocols for Model IP Protection}

\author[1]{Racchit Jain}
\author[1]{Satya Lokam}
\author[1]{Yehonathan Refael\thanks{Author worked on this research during his internship at Microsoft.}}
\author[1]{Adam Hakim}
\author[1]{Lev Greenberg}
\author[1]{Jay Tenenbaum}

\affil[1]{Microsoft}

\usepackage{enumerate}
\usepackage{bbm}
\usepackage{dsfont}

\usepackage{hyperref}
\usepackage{url}

\usepackage{float}
\usepackage{amssymb}
\usepackage{amsmath}
\usepackage{amsthm}
\usepackage{enumitem}
\usepackage{caption}

\usepackage{multirow}
\usepackage{graphicx}

\newtheorem{theorem}{Theorem}
\newtheorem{lem}[theorem]{Lemma}
\newtheorem{definition}{Definition}

\def\int{\displaystyle\mathop {\mbox{\rm int}}}    

\DeclareUnicodeCharacter{2212}{\ensuremath{-}}

\newcommand{\charlie}{\ensuremath{\mathcal{C}} }
\newcommand{\david}{\ensuremath{\mathcal{D}} }

\newcommand{\rjc}[1]{}
\newcommand{\jtc}[1]{}

\begin{document}

\maketitle

\footnotetext[1]{Email addresses: Racchit Jain (\texttt{t-racjain@microsoft.com}), Satya Lokam (\texttt{satya.lokam@microsoft.com}), Yehonathan Refael (\texttt{t-yrefael@microsoft.com}), Adam Hakim (\texttt{adamhakim@microsoft.com}), Lev Greenberg (\texttt{levgreenberg@microsoft.com}), Jay Tenenbaum (\texttt{talaviv@microsoft.com}).}

\begin{abstract}
\jtc{wrote}
Large Language Models (LLMs) represent valuable intellectual property (IP), reflecting significant investments in training data, compute, and expertise. Deploying these models on partially trusted or insecure devices introduces substantial risk of model theft, making it essential to design inference protocols with provable security guarantees.

We present the formal framework and security foundations of SLIP, a hybrid inference protocol that splits model computation between a trusted and an untrusted resource. We define and analyze the key notions of model decomposition and hybrid inference protocols, and introduce formal properties including safety, correctness, efficiency, and t-soundness. We construct secure inference protocols based on additive decompositions of weight matrices, combined with masking and probabilistic verification techniques. We prove that these protocols achieve information-theoretic security against honest-but-curious adversaries, and provide robustness against malicious adversaries with negligible soundness error.

This paper focuses on the theoretical underpinnings of SLIP: precise definitions, formal protocols, and proofs of security. Empirical validation and decomposition heuristics appear in the companion SLIP~\cite{refael2024slipsecuringllmsip} paper. Together, the two works provide a complete account of securing LLM IP via hybrid inference, bridging both practice and theory.
\end{abstract}

\allowdisplaybreaks

\section{Introduction}
\label{sec:intro} 
\jtc{wrote}
The rapid adoption of Large Language Models (LLMs) has raised new concerns around intellectual property (IP) protection \cite{zhao2023survey}. Training such models requires immense resources, and the resulting parameters constitute valuable proprietary assets \cite{perrault2024artificial}. While cloud deployment offers strong protections, economic pressures increasingly drive model owners to offload computation to less secure environments such as edge devices, where the risk of IP theft is acute \cite{businessinsiderChatGPTCould}.

A natural solution is to split computation between a trusted entity (e.g., a secure cloud server) and an untrusted but efficient entity (e.g., an edge device), a process we refer to as \textit{model decomposition}, and perform inference via what we call a \textit{hybrid inference protocol} \cite{8763885}. The challenge is to design model decompositions and hybrid inference protocols that are both efficient—delegating most computation to the untrusted entity—and secure—ensuring that in every inference the untrusted entity learns nothing more about the sensitive model parameters beyond learning the model output on the given input \cite{petitcolas2023kerckhoffs}.

In our companion work~\cite{refael2024slipsecuringllmsip}, we introduced SLIP, a practical hybrid inference method for LLMs, and demonstrated empirically that it preserves accuracy and latency while significantly mitigating model theft. That paper focused on the machine learning aspects of SLIP: decomposition strategies via Singular Value Decomposition (SVD) \cite{meng2024pissa,wang2024svd,sharma2023truth}, and empirical evaluation.

In this paper, we develop the formal security foundations of SLIP. Specifically, we
\begin{enumerate}
    \item Define \textbf{model decomposition} formally and introduce the notion of \textit{safety}, which ensures that the untrusted party’s partial model does not enable IP extraction beyond standard black-box attacks \cite{birch2023model,shamir2023polynomial};

    \item Formalize \textbf{hybrid inference protocols}, and define properties including \textit{correctness, efficiency, security,} and \textit{t-integrity} against malicious adversaries \cite{cryptoeprint:2023/1147};

    \item Construct secure protocols for both the honest-but-curious and malicious settings, proving information-theoretic security in the former case and negligible soundness error in the latter \cite{abadi2016deep,knott2021crypten}.
\end{enumerate}

Our results establish SLIP as the first hybrid inference protocol for LLMs with rigorous and provable security guarantees. By separating the contributions into two companion papers: this work for formal proofs, and SLIP~\cite{refael2024slipsecuringllmsip} for empirical validation, we aim to provide the security and cryptography community with both a solid theoretical foundation and evidence of practical feasibility.

\section{General Framework: Model Decomposition and Hybrid Inference Protocol}
\label{sec:framework_settings}

In this section, we describe a general framework to offload most of the computational workload associated with model inferences to an untrusted entity, while concurrently protecting the model's IP. The two main conceptual ingredients in our framework are model decomposition and the hybrid inference protocol, as shown in Figure \ref{fig:framework}. We first define these notions by considering a \emph{trusted} entity $\mathcal{C}$harlie and an \emph{untrusted} entity $\mathcal{D}$avid. After a model is decomposed and distributed between them, on any given input, Charlie and David execute a hybrid inference protocol to compute the inference output that the original model would have produced on that input. The \emph{Safety} property of the decomposition that we define in Section~\ref{sec:safe-decomp} guarantees that Charlie's part contains most of the IP of the original model whereas David's part has insignificant information from the original model. Additionally, it is imperative that the hybrid protocol be efficient in the sense that David does the bulk of the computation to run inference through the protocol and secure in the sense that it ensures that it prevents the sensitive information pertaining to the model held by Charlie from being leaked to David throughout the protocol's enactment. The required \emph{efficiency} and \emph{security} properties for hybrid inference protocols are defined in Section~\ref{sec:secure-hybrid-inference}. Such an efficient and secure protocol executed on a safe model decomposition then achieves our aim - protecting the IP of the model while offloading most of the computational workload to David.

We start by defining a model decomposition and a hybrid inference protocol. Then, the next two subsections define the desired properties of a model decomposition and a hybrid inference protocol. 

\medskip

\begin{definition}[Model Decomposition]
\label{def:model-decomposition}
Let $\phi_{\Theta}$ be a model with parameters $\Theta$. A \emph{model decomposition} of $\phi_{\Theta}$ is defined by an efficient\footnote{Here, as is standard in cryptography, we can define efficient running time to be probabilistic polynomial time in input size, e.g., number of parameters. But in machine learning, we often rely on more empirical notions of efficiency. For instance, when considering model extraction attacks, we may consider an attacker to be efficient if it uses significantly less resources, e.g., FLOPS, than a brute force re-training of the model from scratch.}  decomposition algorithm $\mathsf{decompose} \: : \: \phi_{\Theta} \: \mapsto \: (\phi_{\Theta_1}, \phi_{\Theta_2})$  that takes as input $\phi_{\Theta}$ and  outputs  a pair of (partial) models $\phi_{\Theta_1}$ and $\phi_{\Theta_2}$. 

We often omit the model descriptors $\phi, \phi_{\Theta_1}, \phi_{\Theta_2}$ when they are clear from specifications of their respective parameters sets $\Theta, \Theta_1, \Theta_2$ and simply denote the decomposition as $\Theta = (\Theta_1, \Theta_2)$. 
\end{definition}

Note that $\phi_{\Theta_i}$, $i=1,2$, need not be fully functional models by themselves, e.g., in the case of a neural network comprised of fully connected layers, i.e., Multi Layer Perceptrons (MLP's), they could be disconnected directed weighted graphs with nodes defining neural operations. 

\medskip

\begin{definition}[Hybrid Inference Protocol]
\label{def:hybrid-inf}
Given a model decomposition $\Theta = (\Theta_{\charlie}, \Theta_{\david})$, where $\Theta_{\charlie}$ is given to Charlie and $\Theta_{\david}$ is given to David, 
a Hybrid Inference protocol $\Pi := \Pi(\Theta_{\charlie}, \Theta_{\david})$ is a communication protocol between Charlie and David that proceeds as follows:

\begin{itemize}
    \item \textbf{Initialization:} \emph{the input $x$ for inference is given to both parties. 
    Charlie’s local state is initialized to include $\Theta_{\charlie}$; 
    similarly David’s local state is initialized to include $\Theta_{\david}$. 
    Protocol $\Pi$ may also specify an initialization or \emph{set up} procedure that the parties must run 
    to populate their initial states and compute parameters needed to compute and interpret messages in the protocol.}

    \item \textbf{Communication:} \emph{Charlie and David communicate by turns, where each party computes its next message 
    and updates its local state as a (local) function of the history of their communication so far and its local state. 
    The protocol $\Pi$ specifies these local functions. We also assume a perfect communication channel between them, 
    i.e., messages are always delivered and never tampered.}

    \item \textbf{Output:} \emph{At the end of the protocol, both parties output either a value $\Pi(x)$ or Charlie outputs the abort symbol $\bot$.}
\end{itemize}
\end{definition}

\begin{figure}[h]
    \centering
    \includegraphics{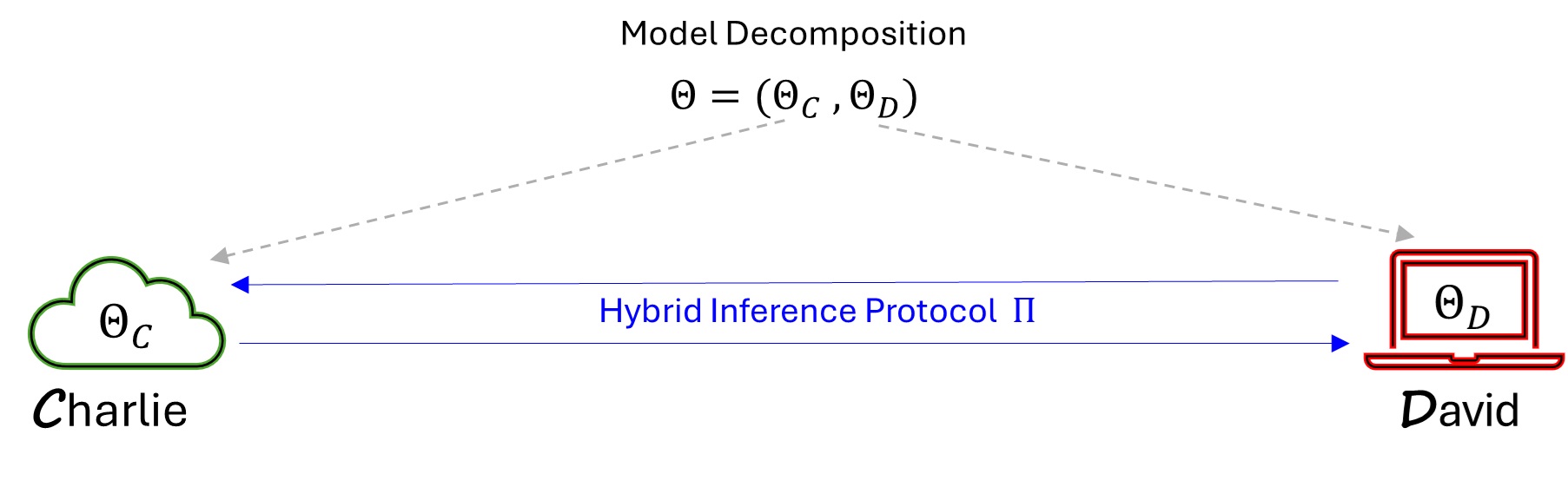}
    \caption{The proposed framework, consisting of the Model Decomposition and the Hybrid Inference Protocol}
    \label{fig:framework} 
\end{figure}

\subsection{Properties of Model Decomposition}
\label{sec:safe-decomp}

In this section, we define properties of a model decomposition that formalize its ability to protect the IP of the model while offloading most of the computational cost of model inference to an untrusted entity. To formalize the notion of model IP, we use the risk of a model.

For a task data distribution $P(x,y),$ a model $\mathcal{M}$ and a loss function $\ell$, the model true risk is given by $\mathsf{risk}(\mathcal{M}) := \mathbb{E}_P \left[\ell(y, \mathcal{M}(x))\right].$ 
Let $A$ be an algorithm that, given the partial model $\Theta_{\david}$ tries to reconstruct an approximation $\hat{\Theta}$ to $\Theta$. The intuition is that the algorithm $A$ is an attacker trying to "reverse engineer" $\Theta$ given $\Theta_{\david}$ model, i.e., steal the IP using information given to David. 

\medskip
\begin{definition}[Safety]
\label{def:safety} 
A model $\mathcal{M}_1$ is an \emph{$\epsilon$-approximation} to a model $\mathcal{M}_2$, if $\mathsf{risk}(\mathcal{M}_1) \leq (1+ \epsilon) \cdot \mathsf{risk}(\mathcal{M}_2)$, $0 \leq \epsilon \leq 1$; the smaller the $\epsilon$, the better the approximation.
A decomposition $\Theta = (\Theta_\charlie, \Theta_\david)$ is said to be \emph{safe} if,
for every adversary ${\david}_1$, 
given $\Theta_{\david}$ and black box access to $\Theta$,
recovers $\epsilon_1$-approximation $\Theta_1$  to $\Theta$ in probabilistic time, polynomial in number of parameters $|\Theta|$ of $\Theta$,  
there is a probabilistic polynomial time adversary ${\david}_2$ that can construct $\epsilon_2$ approximation $\Theta_2$ to $\Theta$ using \emph{only} black box queries to $\Theta$ (i.e., without using $\Theta_{\david}$) for some $\epsilon_2 \approx \epsilon_1$. See Figure~\ref{fig:defn-safety} for an illustration. 
\end{definition}

\begin{figure}[htb]
    \centering
    \includegraphics[width=0.8\linewidth]
    {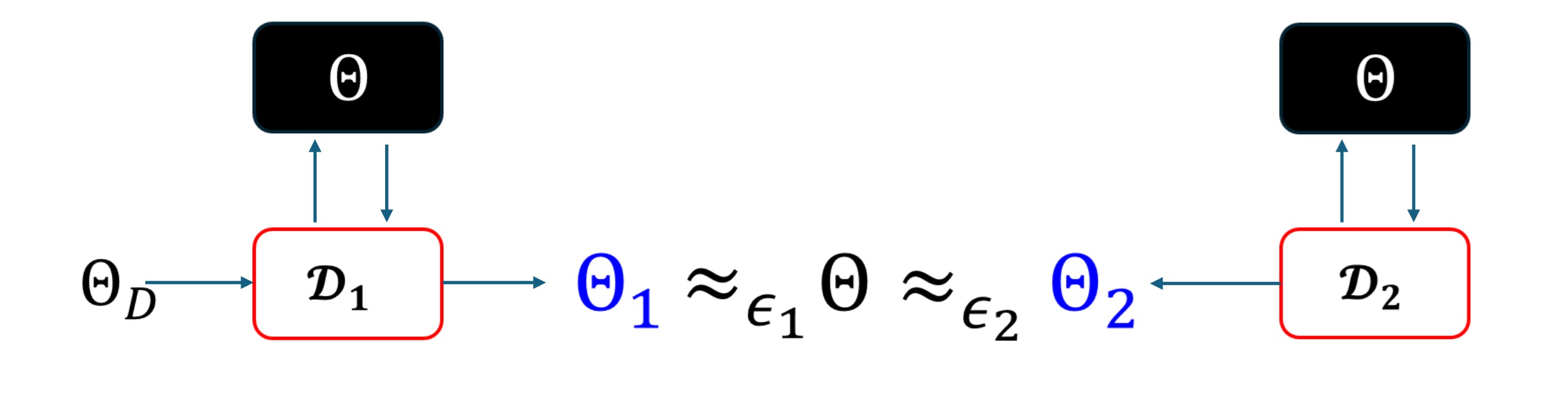}
    \caption{Safety of a Model Decomposition.}
    \label{fig:defn-safety}
\end{figure}

In more detail, ${\david}_1$ is an algorithm that computes the $\epsilon_1$ approximation $\Theta_1$ to $\Theta$. ${\david}_1$ has black box (oracle) access to $\Theta$ and access to $\Theta_{\david}$. ${\david}_2$ is an algorithm that tries to construct an $\epsilon_2$ approximation $\Theta_2$ to $\Theta$ having only access to black box $\Theta$. Note that such a $\david_2$ is meaningful in practice, for instance the black box attack by Shamir et al.\cite{canales2024polynomial} shows that such a $\david_2$ is feasible for certain models $\phi_\Theta$.What the definition says is that essentially such black box attacks are the only means $\david_1$ has to ``reverse engineer'' $\Theta$ even after being given $\Theta_{\david}$; in other words, any extra information (beyond black box access to $\Theta$) that $\david_1$ gets from  $\Theta_D$ is essentially useless. 

Following Definition~\ref{def:model-decomposition}, a model $\phi_{\Theta}$ with parameters $\Theta$ 
is decomposed into two parts $\Theta = (\Theta_{\charlie}, \Theta_{\david})$ using an efficient algorithm 
$\mathsf{decompose}$. In our framework, we focus on \emph{additive decompositions} of the weight matrices 
in $\Theta$. For any layer $i$ with weight matrix $W_i \in \mathbb{Z}_p^{d_{i+1} \times d_{i}}$, we instantiate this 
decomposition as $W_i = W_i^{\charlie} + W_i^{\david}$, where $W_i^{\charlie}$ corresponds to the part of $\Theta_{\charlie}$ retained locally by Charlie and $W_i^{\david}$ corresponds to the part of $\Theta_{\david}$ sent to David.

Although any additive decomposition is theoretically allowed, Refael et al.~\cite{refael2024slipsecuringllmsip} find that the 
\emph{Singular Value Decomposition (SVD)} provides a principled and effective heuristic. 
For a weight matrix $W$, the SVD decomposition is $W = U \Sigma V^\top$,
where $U \in \mathbb{Z_p}^{d_{i+1} \times r}$ and $V \in \mathbb{Z_p}^{d_i \times r}$ are orthonormal matrices and $\Sigma = \mathrm{diag}(\sigma_1, \sigma_2, \dots, \sigma_r)$ contains singular values sorted in decreasing order: $\sigma_1 \geq \sigma_2 \geq \dots \geq \sigma_r$.

We define an SVD-based additive split by retaining the top-$k$ singular components in $W^{\charlie}$ and assigning 
the rest to $W^{\david}$:
\[
    W^{\charlie} = U_k \Sigma_k V_k^\top,
    \qquad
    W^{\david} = W - W^{\charlie}.
\]

The concept of a safe decomposition is the foundation to protect the IP of a model in our framework. However, it is not sufficient. In order to offload computation to an untrusted entity (David) such a decomposition must be nontrivial and asymmetrically split the computation for inference between Charlie and David. They must communicate using an interactive protocol. Such a protocol is called a hybrid inference protocol and its desired properties are formalized next. 

\subsection{Properties of Hybrid Inference Protocol}
\label{sec:secure-hybrid-inference}
Recall from Definition~\ref{def:hybrid-inf} that at the beginning of the hybrid inference protocol, $\Pi(\Theta_{\charlie}, \Theta_{\david})$, $\charlie$ and $\david$ both receive the input $x$,  and at the end of executing $\Pi$, they should both obtain the output of the protocol $\Pi(x)$, which may be equal $\Theta(x)$ or $\bot$. Intuitively, we want a protocol where the computation done by Charlie is much smaller than if Charlie were to run the inference on the entire model by itself. We also want the protocol to output $\Pi(x) = \Theta(x)$ when both parties behave honestly. Since we offload $\Theta_{\charlie}$ to David and it performs most of the computation, it then becomes necessary that David's participation in the protocol should not leak any extra information about Charlie's part of the model than what can be extracted through black box queries to $\Theta$. Basically while participating in the protocol David should not learn anything beyond the input to the model, the output of the model and the part of the model that is owned by David ($\Theta_{\david}$). 

To formalize this notion of "not learning anything beyond the input and output" for David, we define two \emph{views} of David. Let $\mathsf{View_{\Pi}^{\mathsf{real}}}(x)$ be David's \emph{real world} view while participating in the protocol $\Pi$ with Charlie on some input $x$. David's real world view consists of its local state and messages received by Charlie during the execution of the protocol $\Pi$. Let $\mathsf{View_{Sim}^{\mathsf{ideal}}}(x)$ be David's \emph{ideal world} view on some input $x$. We define a simulator $\mathsf{Sim}(x, \Theta(x), \Theta_{\david})$ that is only given the input to the model $x$, the output of the model $\Theta(x)$ and David's part of the model $\Theta_{\david}$. In the ideal world David knows nothing beyond the model input, output and the split of the model in its local state. The Simulator's job is to construct an ideal world view for David that is indistinguishable from his real world view. If such a simulator exists, that means that whatever David can compute in the real protocol he could have computed it just from $x, \Theta(x)$. Therefore the existence of this simulator implies that David learns nothing beyond the input and output of the model while participating in $\Pi$.

In the case where David does not follow the protocol honestly and acts as a malicious adversary, the protocol should be designed in a way such that Charlie doesn't accept an incorrect output except with negligible probability. These properties are captured in the following definitions. 

\medskip

\begin{definition}[Efficiency]
\label{defn:eff-decomp} 
For a hybrid inference protocol $\Pi$ (Definition~\ref{def:hybrid-inf}), let $T_{\Pi, \charlie}(\Theta_{\charlie}, \Theta_{\david})$ be the time complexity of \emph{local} computation by \charlie during his participation in $\Pi$. Let $T_{\charlie}(\Theta)$ be the time complexity of Charlie for running inference on $\phi_\Theta$ all by himself. A hybrid inference protocol $\Pi(\Theta_\charlie, \Theta_\david)$ is said to be \emph{efficient}, if, for a sufficiently small $\epsilon$ (say, $\epsilon = 0.1$),  
\begin{align}
\label{eq:eff-decomp}
T_{\Pi,\charlie}(\Theta_{\charlie}, \Theta_{\david}) & \; \leq  \; \epsilon \cdot T_{\charlie}(\Theta).
\end{align}

\end{definition} 

\begin{definition}[Security]
\label{defn:Security-decomp} 
A inference protocol $\Pi := \Pi(\Theta_{\charlie}, \Theta_{\david})$ between $\charlie$harlie and $\david$avid on a decomposition $ \Theta = ({\Theta_\charlie}, {\Theta_\david})$ is said to be \emph{secure} if, for every valid inference input $x$, there exists a probabilistic polynomial time simulator $\mathsf{Sim}(x, \Theta(x), \Theta_{\david})$ such that $$\mathsf{View_{\Pi}^{\mathsf{real}}}(x, \Theta(x)) \approx_c \mathsf{View_{Sim}^{\mathsf{ideal}}} (x,\Theta(x))$$

(see figure~\ref{fig:defn-security}).\jtc{I see this is the same definition. How does this work with the proof?}
\end{definition}

\begin{definition}[Correctness]
\label{defn:correct} 
A hybrid inference protocol $\Pi := \Pi(\Theta_{\charlie}, \Theta_{\david})$ between $\charlie$harlie and $\david$avid on a decomposition $ \Theta = ({\Theta_\charlie}, {\Theta_\david})$ is said to be correct if for every valid inference input $x$, when both parties execute the protocol honestly, i.e, there is no abort, $\Pi(x) = \Theta(x)$.
\end{definition} 

\begin{definition}[$t$-Integrity]
\label{defn:sound}
A hybrid inference protocol $\Pi := \Pi(\Theta_{\charlie}, \Theta_{\david})$ 
between Charlie and David on a decomposition 
$\Theta = (\Theta_{\charlie}, \Theta_{\david})$ achieves \emph{$t$-integrity}
against a malicious David if for every probabilistic polynomial-time adversary acting as David and for every input $x$, 
\[
    \Pr\!\left[\, \Pi(x) \neq \Theta(x) \mid  \Pi(x) \neq \bot \,\right] \le t,
\]
where the probability is taken over the internal randomness of the protocol and the adversary.
\end{definition}

\begin{figure}[htb]
    \centering
    \includegraphics[width=0.8\linewidth]
    {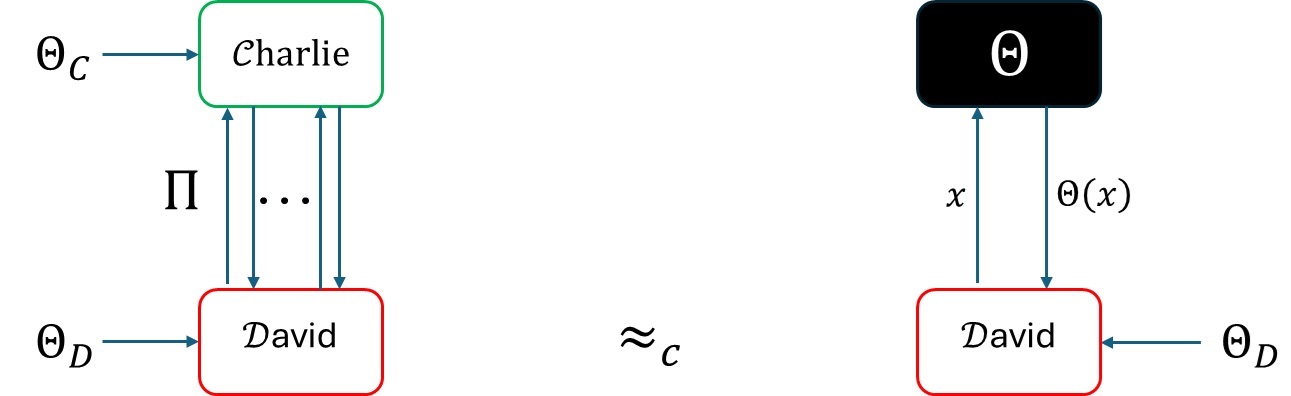}
    \caption{Security of a Hybird Inference Protocol $\Pi$.}
    \label{fig:defn-security}
\end{figure}



\section{Secure Hybrid Inference Protocol}
\label{sec:protocol}
 Before discussing the protocol, let us note that a theoretically obvious approach to offload inference to an untrusted party while protecting the model is to use \emph{obfuscation} that is, hiding the model weights cryptographically while preserving model functionality. However, despite astonishing breakthroughs in recent years \cite{CACM-iO},
 current solutions for cryptographically secure obfuscation are extremely impractical even for simple functions. For instance, obfuscating the simple AND function on 32 bits is estimated to take 23.5 minutes with the obfuscated program taking 11.6 GB to store and 77 seconds to evaluate ~\cite{CGM+-Obfus}. We thus need more efficient alternatives that can protect functionalities represented by modern DNN models. 
 
 Our approach here addresses this challenge by decomposing the model between a trusted entity Charlie and an untrusted computing entity David and designing a protocol between them for inference. This protocol will be secure and efficient (cf. Definitions \ref{defn:eff-decomp} and \ref{defn:Security-decomp}).

\jtc{Is this the best place to mention the limitations of our protocol?}We demonstrate the protocol on MLPs, however it is applicable on all types of layers that have a linear operation on the input and a non linear operation followed by it.\jtc{In the ML paper I made it more specific - Matrix-vector multiplications in general.}Recall that the decomposition $\Theta = (\Theta_{\charlie}, \Theta_{\david})$ splits weight matrices of the model's layers such that it satisfies the \emph{safe decomposition} property.

\jtc{I would place this closer to where you use it. The reader has no context here yet.}We assume that all arithmetic in the protocol is done over a finite field $\mathbb{F} = \mathbb{Z}_p$ for a large prime $p$. That is, all operations are done modulo a prime $p$ that is large enough such that the they avoid wrap-arounds so as to recover the final output without loss of precision. 

\subsection{$\Pi_{\mathsf{honest}}$: Secure Inference Against Honest But Curious David}
\label{sec:single-MLP}
We now present a protocol to securely compute inference on an $L$ layer MLP under the assumption that David, the untrusted party is an \emph{honest but curious} adversary. For illustration purposes, we will describe the protocol for a single layer $i$. However, the protocol proceeds for all the layers in a similar manner and the arguments are composable. Let $W_i \in \mathbb{Z}_p^{d_{i+1} \times d_i}$ be the weight matrix of the $i$-th layer and $a_i \in \mathbb{Z}_p^{d_{i+1}}$ be the output of the activation function ($\sigma$) after $W_i$ is applied to its input $(a_{i-1} \in \mathbb{Z}_p^{d_i})$, the output of the previous layer (linear + activation). In other words, 
\begin{align}
\label{eq:basic-layer}
a_i & = \sigma (W_i a_{i-1}),
\end{align} 
for $i \geq 1$, where we take $a_0 = x$, the input to the model $\Theta$. For simplicity, we assume that $x$ is available initially to both $\charlie$ and $\david$ and that at the end of the inference, the output $\Theta(x)$ should also available to both parties. Let us first look at what a naive insecure inference protocol might look like.  

\textbf{Insecure inference protocol.} For a layer $i$, the decomposition $\Theta = (\Theta_{\charlie}, \Theta_{\david})$ splits $W_i = W_i^{\charlie} + W_i^{\david}$. Thus an \emph{insecure} protocol, where Charlie's part of the model $\Theta_{\charlie}$ is not being protected, is: 
\begin{enumerate}
\item[(i)] Charlie $\charlie$ sends $a_{i-1}$ to David $\david$ and locally computes $a_i^{\charlie} := W_i^{\charlie} a_{i-1}$,  
\item[(ii)] $\david$ computes $a_i^{\david} := W_i^{\david} a_{i-1}$ and sends $a_i^{\david}$ back to Charlie.
\item[(iii)] Charlie computes the output of the layer $a_i$ according to equation~\ref{eq:insecure}.
\begin{align}
\label{eq:insecure}
a_i & = \sigma (a_i^{\charlie} + a_i^{\david}) = \sigma ((W_i^{\charlie} + W_i^{\david}) a_{i-1}) = \sigma (W_i a_{i-1}). 
\end{align}
\item[(iv)] Now, Charlie and David are ready for the execution of the $(i+1)$-st layer (linear + activation) on its input $a_i$. 
\end{enumerate}


What is wrong with the insecure protocol? From the point of view of model protection, in the worst case scenario (l), the activation function (such as ReLU) has no impact (or is invertible, such as with \texttt{softmax}) on the right hand side of \eqref{eq:basic-layer}. Therefore, we can generate a linear equation for the unknowns $W_i$ where $a_{i-1}$ and $a_i$ are known to the adversary $\david$. Thus, making sufficiently many inference queries, $\david$ can set up linear equations $a_i^t = W_i a_{i-1}^{t}$ for sufficiently many $t$ to solve for $W_i$. Thus the trivial protocol above does \emph{not} protect the model. 

\textbf{Secure Inference Protocol $\Pi_{\mathsf{honest}}$.} We now describe a secure hybrid inference protocol $\Pi_{\mathsf{honest}}$  that can be proven to prevent David from stealing the model $\Theta$ based on $\Pi_{\mathsf{honest}}$, \emph{assuming} that $\Theta = (\Theta_C, \Theta_D)$ is a \emph{safe decomposition} according to Definition~\ref{def:safety}. The basic concept is simple, Charlie masks $a_{i-1}$ with a \emph{random one time pad} $r_{i-1}$ sampled uniformly at random from $\mathbb{Z}_p^{d_i}$ and sends $\widetilde{a_{i-1}} = a_{i-1} + r_{i-1}$ to David instead of just $a_{i-1}$ in Step (i) of the insecure protocol above. David then computes as before in Step (ii) and sends an ``incorrect'' $\widetilde{a_i^{\david}}$ to Charlie. This is then cleaned up by Charlie using a \emph{pre-computed} cancellation mask. This allows Charlie to compute $a_i$ as before in Step (iii) and to continue onwards. Note that for $i = 1$, $a_{i-1} = a_0 = x$. Since input is available in plaintext to both Charlie and David we don't need to mask $a_0$ and can send it to David in the clear and hence, don't need to compute cancellation masks or sample random vectors for $i=1$. For $i=1$, the insecure protocol suffices. The secure protocol is described below in figures ~\ref{fig:secure-inference-offline-honest-but-curious} and ~\ref{fig:secure-inference-online-honest-but-curious}. 

\begin{figure}[H]
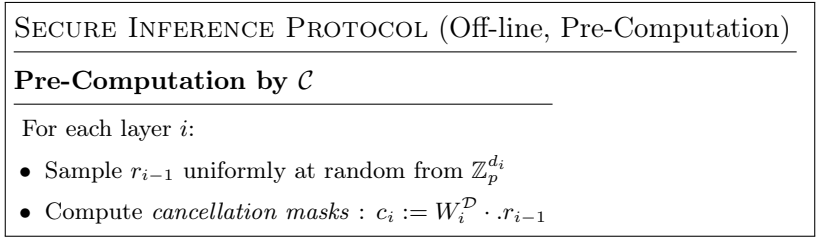

    \centering
     \scalebox{1}{
        \fbox{%
            \procedure{\textsc{Secure Inference Protocol} (Off-line, Pre-Computation) }{%
                {%
                    \begin{subprocedure}\procedure{\textbf{Pre-Computation by $\charlie$}}{%
                        \text{ For each layer $i$: } \\
                        \bullet \text{   Sample }  r_{i-1} \text{ uniformly at random from } \mathbb{Z}_p^{d_{i}} \\
                        \bullet \text{   Compute \emph{cancellation masks} : } c_i := W_i^{\david} \cdot. r_{i-1}
                    }
                    \end{subprocedure}
                } 
            }
            }
            }
    \caption{Secure Inference Protocol -- Pre-Computation}
    \label{fig:secure-inference-offline-honest-but-curious}             
\end{figure}


\begin{figure}[H]
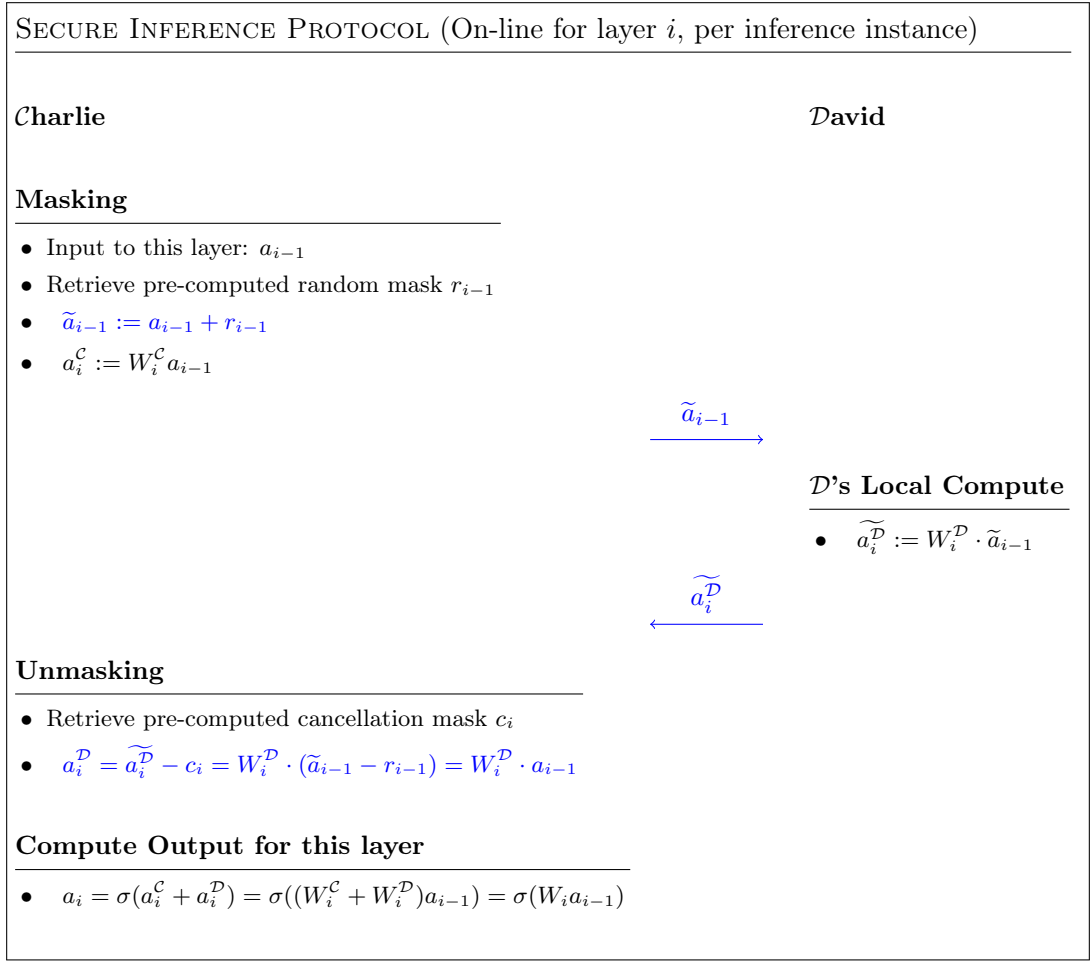

    \centering
     \scalebox{1}{
        \fbox{%
            \procedure{\textsc{Secure Inference Protocol} (On-line for layer $i$,  per inference instance)}{%
                \< \< \\
                \textbf{$\charlie$harlie} \<   \<  \textbf{$\david$avid} \\
                \< \< \\
                {%
                    \begin{subprocedure}\procedure{\textbf{Masking}}{%
                     \bullet \text{  Input to this layer: } a_{i-1} \\
                     \bullet \text{  Retrieve pre-computed random mask } r_{i-1} \\
                     \bullet \quad  \textcolor{blue}{\widetilde{a}_{i-1} := a_{i-1} + r_{i-1}} \\ 
                     \bullet \quad a_i^{\charlie} := W_i^{\charlie} a_{i-1} 
                    }
                    \end{subprocedure}
                } \<  \< \\
                  \< \textcolor{blue}{\sendmessageright*[1.5 cm]{\widetilde{a}_{i-1}}}  \<   \\
                  \<  \< 
                {%
                    \begin{subprocedure}\procedure{\textbf{$\david$'s Local Compute}}{%
                    \bullet \quad \widetilde{a_i^{\david}} := W_i^{\david} \cdot \widetilde{a}_{i-1}
                    }
                    \end{subprocedure}
                }\\
               \< \textcolor{blue}{\sendmessageleft*[1.5 cm]{\widetilde{a_i^{\david}}}} \<   \\
               {%
                    \begin{subprocedure}\procedure{\textbf{Unmasking}}{%
                    \bullet \text{ Retrieve pre-computed cancellation mask $c_i$ } \\
                    \bullet \quad \textcolor{blue}{a_i^{\david} = \widetilde{a_i^{\david}} - c_i = W_i^{\david} \cdot \left(\widetilde{a}_{i-1} -  r_{i-1} \right) = W_i^{\david} \cdot a_{i-1}}
                    }
                    \end{subprocedure}
                } \<   \<  \\
                \< \< \\
               {%
                    \begin{subprocedure}\procedure{\textbf{Compute Output for this layer}}{%
                    \bullet  \quad a_i = \sigma (a_i^{\charlie} + a_i^{\david}) = \sigma ((W_i^{\charlie} + W_i^{\david}) a_{i-1}) = \sigma (W_i a_{i-1}) 
                    }
                    \end{subprocedure}
                } \<   \<  \\
            } 
         } 
    }  
    \caption{Secure Inference Protocol -- Online, per Inference Instance}
    \label{fig:secure-inference-online-honest-but-curious}
\end{figure} 

\begin{theorem} 
\label{thm:pf-pPi-single}
The hybrid inference protocol $\Pi_{\mathsf{honest}}$ on decomposition $\Theta = (\Theta_{\charlie}, \Theta_{\david})$, where $\Theta$ is an $L$-layer MLP, is correct and secure against an honest but curious $\david$ for every input x.
\end{theorem}

\begin{proof}
    The protocol satisfies correctness since the output at a layer $i$ if David follows the protocol honestly is $\sigma(W_ia_{i-1})$, which would be $\Pi_{\mathsf{honest}}(x) = a_L = \sigma(W_La_{L-1}) = \Theta(x)$ at the output layer.

    In the real world David and Charlie participate in the protocol $\Pi_{\mathsf{honest}}$  and since the output of the protocol is shared with both parties at the end of the protocol, David has access to the model output as well as intermediate activations. David has the real world view: $$\mathsf{View}_{\Pi_{\mathsf{honest}}}^{\mathsf{real}}(x) = \{ \Theta_{\david}, \{\tilde{a}_i\}_{i=1}^{L}, \Pi_{\mathsf{honest}}\}$$

    The ideal world view is generated by the simulator $\mathsf{Sim}(x, \Theta(x), \Theta_{\david})$. To prove security we have to show the existence of a $\mathsf{Sim}$ that generates an ideal world view that is computationally indistinguishable from David's real world view. 

    Let us construct the simulator such that $\mathsf{Sim}(x,\Theta(x), \Theta_{\david}) = \{\Theta_{\david}, \{u_i\}_{i=1}^{L}, \Theta(x)\}$ where $u_i$ is a truly random vector of the same length as $\tilde{a}_{i}$ with entries sampled independently and uniformly at random from $\mathbb{Z}_p$ and $$\mathsf{View}_{\mathsf{Sim}}^{\mathsf{ideal}}(x) = \mathsf{Sim}(x,\Theta(x), \Theta_{\david}) = \{\Theta_{\david}, \{u_i\}_{i=1}^{L}, \Theta(x)\}$$

    

    Since the messages $\{\widetilde{a}_{i}\}_{i=1}^L$ from Charlie to David in the execution of $\Pi_{\mathsf{honest}}$ are indistinguishable from random vectors, this simulated ideal view of David is indistinguishable from his honest participation in $\Pi$. And when both parties participate honestly, we have shown that the protocol is correct hence $\Pi_{\mathsf{honest}}$ in the real world view is the same as $\Theta(x)$ in the ideal world view generated by $\mathsf{Sim}$. Hence we have shown that the real world view of David's  participation in the protocol is computationally indistinguishable from his ideal world view where he only sees the output of the model $\Theta$. 
    
\end{proof}

 Note that no cryptographic assumptions are used in the proof and hence we can call it ``information-theoretically secure.'' 

\subsection{$\Pi_{\mathsf{malicious}}$: Secure Inference of an MLP against Malicious David}
\label{sec:single-int-MLP}
The protocol described in ~\ref{fig:secure-inference-online} and ~\ref{fig:secure-inference-offline} is secure against honest but curious adversaries. We will now provide a protocol $\Pi_{\mathsf{malicious}}$ for an L layer MLP that is secure against malicious adversaries. This protocol like the previous protocol guarantees that no model IP is revealed to David through participation in the protocol under the assumption that $\Theta = (\Theta_{\charlie}, \Theta_{\david})$ is a \emph{safe decomposition} according to Definition~\ref{def:safety} along with the additional guarantee that if David doesn't participate in the protocol honestly, it will be caught with high probability and the protocol will output abort ($\bot$). 

The protocol operates similarly as seen in figures \ref{fig:secure-inference-online-honest-but-curious} and \ref{fig:secure-inference-offline-honest-but-curious} but with an integrity check. We do a batched integrity check where we sample $k$ uniformly random vectors from $\mathbb{Z}_{p}^{d_{i+1}}$ to build the matrix $Z_{i-1} \in \mathbb{Z}_{p}^{k \times d_{i+1}}$. We then compute our verification matrix $V_{i-1} = Z_{i-1}W_{i}^{\david}$. In the online phase Charlie verifies David's computation by simply doing $k$ inner products and checking if $Z_{i-1} \widetilde{a_{i}^{\david}}$ is equal to $V_{i-1}\tilde{a}_{i-1}$.

We note that for $i =1$, for which $a_0 = x$ which is the input to the model, we can optimize efficiency by not computing noise masks and cancellations masks similar to $\Pi_{\mathsf{honest}}$, since the input is available to both parties at the beginning of the protocol anyways. However we still need to perform an integrity check since David can maliciously output $\widetilde{a_1^{\david}}$. Therefore we need to sample $Z_0$ and compute $V_0$.

\begin{figure}[H]
    \centering
     \scalebox{1}{
        \fbox{%
            \procedure{\textsc{Secure Inference Protocol With Integrity Check} (Off-line, Pre-Computation) }{%
                {%
                    \begin{subprocedure}\procedure{\textbf{Pre-Computation by $\charlie$}}{%
                        \text{ For a layer $i$: } \\
                        \bullet \text{   Sample }  r_{i-1} \text{ uniformly at random from } \mathbb{Z}_p^{d_{i}}\\
                        \bullet \text{   Compute \emph{cancellation masks} : } c_i := W_i^{\david} \cdot. r_{i-1}\\
                        \bullet \text{   Sample } Z_{i-1} \text{ uniformly at random from } \mathbb{Z}_p^{k \times d_{i+1}}\\
                        \bullet \text{  Compute \emph{verification vectors} : } V_{i-1} := Z_{i-1} \cdot W_i^{\david}
                    }
                    \end{subprocedure}
                } 
            }
            }
            }
    \caption{Secure Inference Protocol With Integrity Check -- Pre-Computation}
    \label{fig:secure-inference-offline}             
\end{figure}


\begin{figure}[H]
    \centering
     \scalebox{1}{
        \fbox{%
            \procedure{\textsc{Secure Inference Protocol} (On-line,  per inference instance for layer $i$)}{%
                \< \< \\
                \textbf{$\charlie$harlie} \<   \<  \textbf{$\david$avid} \\
                \< \< \\
                {%
                    \begin{subprocedure}\procedure{\textbf{Masking}}{%
                     \bullet \text{  Input to this layer: } a_{i-1} \\
                     \bullet \text{  Retrieve pre-computed random mask } r_{i-1} \\
                     \bullet \quad  \textcolor{blue}{\widetilde{a}_{i-1} := a_{i-1} + r_{i-1}} \\ 
                     \bullet \quad a_i^{\charlie} := W_i^{\charlie} a_{i-1} \\
                    }
                    \end{subprocedure}
                } \<  \< \\
                  \< \textcolor{blue}{\sendmessageright*[1.5 cm]{\widetilde{a}_{i-1}}}  \<   \\
                  \<  \< 
                {%
                    \begin{subprocedure}\procedure{\textbf{$\david$'s Local Compute}}{%
                    \bullet \quad \widetilde{a_i^{\david}} := W_i^{\david} \cdot \widetilde{a}_{i-1}
                    }
                    \end{subprocedure}
                }\\
               \< \textcolor{blue}{\sendmessageleft*[1.5 cm]{\widetilde{a_i^{\david}}}} \<   \\
               {%
                    \begin{subprocedure}\procedure{\textbf{Unmasking}}{%
                    \bullet \text{ Retrieve pre-sampled random matrix  } Z_{i-1} \\
                    \bullet \text{      and pre-computed verification matrix }  V_{i-1} \\
                    \bullet \textbf{  Linear layer check: } Z_{i-1} \cdot \widetilde{a_i^{\david}} \stackrel{?}{=} V_{i-1} \cdot \widetilde{a}_{i-1} \\
                    \bullet \text{ If the check fails, Charlie rejects the} \\ \quad \text{ computation and aborts outputting $\bot$.} \\
                    \bullet \text{ Otherwise retrieve pre-computed cancellation mask $c_i$ } \\
                    \bullet \quad \textcolor{blue}{a_i^{\david} = \widetilde{a_i^{\david}} - c_i = W_i^{\david} \cdot \left(\widetilde{a}_{i-1} -  r_{i-1} \right) = W_i^{\david} \cdot a_{i-1}}
                    }
                    \end{subprocedure}
                } \<   \<  \\
                \< \< \\
               {%
                    \begin{subprocedure}\procedure{\textbf{Compute Output for this layer}}{%
                    \bullet  \quad a_i = \sigma (a_i^{\charlie} + a_i^{\david}) = \sigma ((W_i^{\charlie} + W_i^{\david}) a_{i-1}) = \sigma (W_i a_{i-1}) 
                    }
                    \end{subprocedure}
                } \<   \<  \\
            } 
         } 
    }  
    \caption{Secure Inference Protocol with Integrity Check -- Online, per Inference Instance}
    \label{fig:secure-inference-online}
\end{figure} 


\begin{lem}[Freivalds' Lemma]
\label{lem:freivalds-integrity}
Let $A$ be an $n \times n$ matrix and let $x, y$ be $n$-dimensional vectors over a field $\mathbb{F}$. Let $r$ be a uniformly random vector from $\mathbb{F}^n$. Suppose $y \neq Ax$. Then:
\[
\Pr[r^Ty = r^T(Ax)] = \frac{1}{|\mathbb{F}|}.
\]
\end{lem}

\begin{proof}
    This lemma is a specific case of the Freivalds' lemma for verfifying matrix-matrix multiplications \cite{Freivalds1977ProbabilisticMC} where the second matrix is an $n \times 1$ matrix. 
\end{proof}

\begin{lem}[Batched Freivalds' Lemma]
\label{lem:batched-freivalds-integrity}
Let \(A \in \mathbb{F}^{n\times n}\) and let \(x,y \in \mathbb{F}^n\) with \(y \neq Ax\).
Draw \(R \in \mathbb{F}^{n\times k}\) whose columns \(r_1,\ldots,r_k\) are independent
and uniformly random from \(\mathbb{F}^n\). Then
\[
\Pr [R^{T} y \;=\; R^{T} A x] = \frac{1}{|\mathbb{F}|^{k}}.
\]
\end{lem}

\begin{proof}
Let \(d := y - Ax \neq 0\). Let \(R = [\,r_1\,|\,\cdots\,|\,r_k\,]\in\mathbb{F}^{n\times k}\)
have independent columns $r_i$ sampled uniformly and independently from $\mathbb{F}^n$.
Define events $E_i := \{ r_i^{T} y = r_i^{T} A x \}$.

By Lemma~2 for each $i$,
\[
\Pr[E_i] = \Pr[r_i^{T} y = r_i^{T} (A x)]
= \frac{1}{|\mathbb{F}|}.
\]
Since \(r_1,\dots,r_k\) are independent, the events \(E_1,\dots,E_k\) are independent. Hence
\[
\Pr[R^{T}y = R^{T}Ax]
=\Pr[\bigwedge_{i=1}^{k} E_i]
=\prod_{i=1}^{k} \Pr[E_i]
=\frac{1}{|\mathbb{F}|^k}.
\]
\end{proof}

\begin{theorem}
\label{thm:pf-pPi-single-integrity}
    The hybrid inference protocol $\Pi_{\mathsf{malicious}}$ on decomposition $\Theta = (\Theta_{\charlie}, \Theta_{\david})$, where $\Theta$ is an $L$-layer MLP, is correct, secure and satisfies t-Soundness with $t = \frac{L}{|\mathbb{F}|^k}$ for every input $x$, in other words,
\end{theorem}

\begin{proof}
    Let us first start with \emph{correctness} and \emph{security} if both parties participate in the protocol honestly. The protocol $\Pi_{\mathsf{malicious}}$ computations which lead to its output are identical to that of $\Pi_{\mathsf{honest}}$, that is, it does equivalent computations beyond the integrity check as illustrated in figures \ref{fig:secure-inference-offline-honest-but-curious} and \ref{fig:secure-inference-online-honest-but-curious}. Hence the proof of correctness and security in the proof of Theorem \ref{thm:pf-pPi-single} suffices. 

    To prove security against a malicious David, note that if David provides an incorrect partial intermediate activation $\widetilde{a_i^{\david}}$, the protocol's output will most likely be incorrect. We aim to demonstrate that such a malicious David gains no more advantage than if he were to follow the protocol and have black-box access to the model. Instead of supplying the correct input to Charlie, a malicious adversary might provide alternative values with the intent to extract Charlie's parameters. However, even if David behaves maliciously, Charlie still provides him with outputs that are indistinguishable from uniform random noise.  

    We now move on to \emph{t-Integrity} against a malicious David. At a layer $i$, a malicious David can adaptively send incorrect $\widetilde{a_{i}^{\david}}$. Charlie performs a linear layer check $Z_{i-1} \cdot \widetilde{a_i^{\david}} \stackrel{?}{=} V_{i-1} \cdot \widetilde{a}_{i-1}$  where $Z_{i-1}$ is a matrix of k vectors with each entry sampled independently and uniformly at random from $\mathbb{F}$ and $V_{i-1} := Z_{i-1} \cdot W_i^{\david}$ is a precomputed verification vector. If $\widetilde{a_{i}^{\david}}$ is incorrect, then $\widetilde{a_{i}^{\david}} \neq W_i^{\david}\widetilde{a}_{i-1}$, so by the batched Freivalds' lemma \ref{lem:batched-freivalds-integrity} we have that $$\text{Pr}[Z_{i-1}\widetilde{a_{i}^{\david}} = Z_{i-1}(W_i^{\david}\widetilde{a}_{i-1})] = 1/|\mathbb{F}|^k.$$ 

    Hence the probability that Charlie accepts an incorrect output at layer $i$ is at most $1/|\mathbb{F}|^k$. 

    \jtc{We discussed this should be F to the k, right?, without the L, since there is no union bound, just if david attacked somewhere, then consider the first layer, and we catch that layer with high probabiliy, and that's it.}
    By a simple union bound we have that,
    $$\text{Pr}[\Pi_{\mathsf{malicious}}(x) \neq \Theta(x) | \Pi_{\mathsf{malicious}}(x)  \neq \bot] \leq L/|\mathbb{F}|^k$$
    
    Hence the protocol achieves $t$-Soundness across the entire $L$-layer inference with overall failure probability at most
    $L / |\mathbb{F}|^k$.

\end{proof}

\section{Discussion}
\label{sec:Discussion}
\jtc{wrote}
This work advances the state of secure machine learning by providing the first hybrid inference protocol for LLMs with formal, information-theoretic security guarantees. SLIP~\cite{refael2024slipsecuringllmsip} demonstrates that it is possible to offload the majority of inference computation to untrusted devices without exposing model IP beyond what is accessible via black-box queries. This is achieved through a principled combination of model decomposition and cryptographically inspired protocol design which we rigorize and provide a security analysis.

Several important considerations arise for the security community. First, while our protocols are provably secure under the defined adversarial models, real-world deployments must address implementation-level threats such as side-channel attacks, weak randomness, and network-level vulnerabilities. Next, our approach complements, but does not replace, other defenses such as trusted and confidential hardware and compute. This solution, in tandem with others remains a promising avenue for secure and private solutions.

We hope that SLIP’s formalization of hybrid inference security will catalyze further research at the intersection of cryptography, systems, and machine learning, and encourage the adoption of provable security standards in practical AI deployments.

\medskip

\bibliographystyle{abbrv}
\bibliography{main.bib}


\end{document}